\documentclass[envcountsame]{llncs}
\usepackage{amsmath,latexsym}
\usepackage[pdftex]{graphicx}
\usepackage{alltt}
\usepackage{algorithmicx}
\usepackage[noend]{algpseudocode}

\renewcommand{\thepart}{\Roman{part}}
\newtheorem{thrm}{Theorem}
\newtheorem{lemm}[thrm]{Lemma}

\newtheorem{remk}[thrm]{Remark}
\newtheorem{defn}[thrm]{Definition}
\newtheorem{cor}[thrm]{Corollary}

\def\itbf#1{\textit{\textbf{#1}}}

\def\s#1{\mbox{\boldmath $#1$}}

\def\pref#1{\mbox{pref(\s{#1})}}
\def\suff#1{\mbox{suff(\s{#1})}}

\def\match{\approx}

\def\:{\mbox{\ :\ }}
\def\floor#1{\lfloor #1 \rfloor}

\def\+{\!+\!}
\def\-{\!-\!}

\def\bproc{{\bf procedure\ }}

\def\bfor{{\bf for\ }}
\def\bto{{\bf to\ }}

\def\bwhile{{\bf while\ }}

\def\band{{\bf and\ }}

\def\bdo{{\bf do\ }}
\def\bif{{\bf if\ }}
\def\bthen{{\bf then\ }}
\def\belse{{\bf else\ }}

\def\la{\leftarrow}

\def\qq{\qquad}
\def\com#1{\hspace{24pt}{\bf $\triangleright$}\hspace{6pt}{\sl #1}}

\def\pref(#1,#2){$#1$ is a prefix of $#2$}
\def\suff(#1,#2){$#1$ is a suffix of $#2$}

\def\reg(#1,#2){$#2$ is $#1$-regular}
\def\notreg(#1,#2){$#2$ is not $#1$-regular}

\begin{document}

\title{Indeterminate Strings, Prefix Arrays \& Undirected Graphs}

\author{Manolis Christodoulakis\inst{1}
\and P.\ J.\ Ryan\inst{2}
\and W.\ F.\ Smyth\inst{2,}\thanks{The work of the third author was supported in part by a grant from the Natural Sciences \& Engineering Research Council of Canada.}
\and Shu Wang\inst{3}
}

\institute{$\!^1\ $Department of Electrical \& Computer Engineering \\
University of Cyprus, PO Box 20537, 1687 Nicosia, Cyprus \\
\email{christodoulakis.manolis@ucy.ac.cy} \\
{\ } \\
$\!^2\ $Algorithms Research Group, Department of Computing \& Software \\
McMaster University, Hamilton, Ontario, Canada\ L8S 4K1 \\
\email{\{ryanpj,smyth\}@mcmaster.ca} \\
\texttt{www.cas.mcmaster.ca/cas/research/algorithms.htm} \\
% {\ } \\
% $\!^3\ $Digital Ecosystems \& Business Intelligence Institute \\
% Curtin University, GPO Box U1987, Perth WA 6845, Australia \\
% \email{W.Smyth@curtin.edu.au} \\
{\ } \\
$\!^3\ $IBM Toronto Software Lab \\
8200 Warden Avenue, Markham, Ontario, Canada\ L6G 1C7 \\
\email{wangs@ca.ibm.com} \\
}

\maketitle

\begin{center}
\today
\end{center}

\begin{abstract}
An integer array $\s{y} = \s{y}[1..n]$ is said to be \itbf{feasible}
if and only if $\s{y}[1] = n$ and, for every $i \in 2..n$,
$i \le i\+ \s{y}[i] \le n\+ 1$.
A string is said to be \itbf{indeterminate} if and only if at least one of its elements
is a subset of cardinality greater than one
of a given alphabet $\Sigma$;
otherwise it is said to be \itbf{regular}.
A feasible array \s{y} is said to be \itbf{regular}
if and only if it is the prefix array of some regular string.
We show using a graph model that every feasible array of integers
is a prefix array of some (indeterminate or regular) string,
and for regular strings corresponding to \s{y},
we use the model to provide a lower bound on the alphabet size.
We show further that there is a 1--1 correspondence between
labelled simple graphs and indeterminate strings,
and we show how to determine the minimum alphabet size $\sigma$
of an indeterminate string \s{x} based on its
\itbf{associated graph} $\mathcal{G}_{\s{x}}$.
Thus, in this sense, indeterminate strings are a more natural object of combinatorial interest
than the strings on elements of $\Sigma$ that have traditionally been studied.
\end{abstract}


\section{Introduction}
\label{sect-intro}
Traditionally, a string is a sequence of letters taken from some alphabet
$\Sigma$.  Since we discuss ``indeterminate strings" in this paper,
we begin by generalizing the definition as follows:
\begin{defn}
\label{defn-string}
A \itbf{string} with base alphabet $\Sigma$ is either empty
or else a sequence of nonempty subsets of $\Sigma$.
A 1-element subset of $\Sigma$ is called a \itbf{regular} letter;
otherwise it is \itbf{indeterminate}.
Similarly, a nonempty string consisting only of regular letters is \itbf{regular},
otherwise \itbf{indeterminate}.
The empty string \s{\varepsilon} is regular.
\end{defn}
All alphabets and all strings discussed in this paper are finite.
We denote by $\Sigma'$ the set of all nonempty subsets of $\Sigma$,
with $\sigma = |\Sigma|$ and $\sigma' = |\Sigma'| = 2^{\sigma}\- 1$.
On a given alphabet $\Sigma$, there are altogether $(\sigma')^n$ distinct nonempty strings
of length $n$,
of which $\sigma^n$ are regular.

% by Manolis: to be checked
%Let $d_\lambda$ denote the \itbf{cardinality} of a letter $\lambda$, that is the number of regular symbols contained in $\lambda$.
%\marginpar{M: new definition}
%Thus, $\lambda$ is regular if and only if $d_\lambda = 1$.
% I removed this remark because the first sentence didn't seem quite
% correct, and I wasn't sure what the other two sentences added.
% \begin{remk}
% \label{rem-basicdefs}
% A string in our sense is just a string in the traditional sense over the
% alphabet $\Sigma'$.  Regular strings may be identified with
% strings over $\Sigma$ in the traditional sense.  Indeterminate strings
% may be thought of as strings for which some letters
% are ambiguously specified.
% \end{remk}
Indeterminate strings were first introduced in a famous paper
by Fischer \& Paterson \cite{FP74},
then later studied by Abrahamson \cite{A87}.
In the last ten years or so, much work has been done
by Blanchet-Sadri and her associates
(for example, \cite{BSH02}) on ``strings with holes'' --- that is,
strings on an alphabet $\Sigma$ augmented by a single letter
consisting of the $\sigma$-element subset of $\Sigma$.
The monograph \cite{B08} summarizes much of the pioneering work in this area.
For indeterminate strings in their full generality,
the third and fourth authors of this paper have collaborated
on several papers,
especially in the contexts of pattern-matching \cite{HS03,HSW06,HSW08,SW09}
and extensions to periodicity \cite{SW08,SW09a}.

\begin{defn}
\label{defn-match}
Two elements $\lambda,\mu$ of $\Sigma'$ are said to \itbf{match}
(written $\lambda \match \mu$) if they have nonempty intersection.
Two strings \s{x}, \s{y} \itbf{match} ($\s{x} \match \s{y}$)
if they have the same length and all corresponding letters match.
\end{defn}
Thus two regular letters match if and only if they are equal.
But note that for indeterminate letters $\lambda, \mu, \nu$,
it may be that $\lambda \match \mu$ and $\lambda \match \nu$,
while $\mu \not\match \nu$:
think $\lambda = \{1,2\}, \mu = 1, \nu = 2$.

\begin{defn}
\label{defn-border}
If a string \s{x} can be written $\s{x} = \s{u_1v}$ and $\s{x} = \s{wu_2}$
for nonempty strings \s{v}, \s{w}, where $\s{u_1} \match \s{u_2}$,
then \s{x} is said to have a \itbf{border} of length $|\s{u_1}| = |\s{u_2}|$.
\end{defn}
Note that choosing $\s{v} = \s{w} = \s{x}$ yields the empty border
\s{\varepsilon} of length $0$.

The \itbf{border array} of a string $\s{x} = \s{x}[1..n]$ is an integer array
$\s{\beta}[1..n]$ such that $\s{\beta}[i]$ is the length of the longest
border of $\s{x}[1..i]$.
For regular strings \s{x}, the border array has the desirable property,
used in pattern-matching algorithms for more than 40 years \cite{MP70},
that any border of a border of \s{x} is also a border of \s{x} --- thus
\s{\beta} actually specifies every border of every prefix of \s{x}.
For indeterminate strings, however,
due to the intransitivity of the match operation, this is not true \cite{SW09,SW09a};
for example,
\begin{equation}
\label{noborder}
\s{u} = a\{a,b\}b
\end{equation}
has a border of length 2 ($a\{a,b\} \match \{a,b\}b$),
and both borders $a\{a,b\}$ and $\{a,b\}b$ have a border of length 1
($a \match \{a,b\}$ and $\{a,b\} \match b$, respectively),
but \s{u} has no border of length 1.
To make sense of such situations, the ``prefix array'' becomes important:

\begin{defn}
\label{defn-prefix}
The \itbf{prefix array} of a string $\s{x} = \s{x}[1..n]$ is the integer array
$\s{y} = \s{y}[1..n]$ such that for every $i \in 1..n$, $\s{y}[i]$ is the
length of the longest prefix of $\s{x}[i..n]$ that matches a prefix of $\s{x}$.
Thus for every prefix array \s{y}, $\s{y}[1] = n$.
\end{defn}
Apparently the first algorithm for computing the prefix array
occurred as a routine in the repetitions algorithm of Main \& Lorentz \cite{ML84};
see also \cite[pp.\ 340--347]{S03}.
A slightly improved algorithm is given in \cite[Section 8.4]{L05},
and two algorithms for computing a ``compressed'' prefix array are
described in \cite{SW08}.

For regular strings the border array and the prefix array are equivalent:
it is claimed in \cite{CHL01,CHL07}, and not difficult to verify, that
there are $\Theta(n)$-time algorithms to compute one from the other.
On the other hand, as shown in \cite{SW08}, for indeterminate strings
the prefix array actually allows all borders of every prefix to be specified,
while the border array does not
\cite{HS03,IMMP03}.
Thus the prefix array provides a more compact and more general mechanism for
identifying borders, hence for describing periodicity, in indeterminate strings.
In the above example (\ref{noborder}), the prefix array of \s{u}
is $\s{y} = 320$, telling us that $\s{u}[2..3] \match \s{u}[1..2]$
(\s{u} has a border of length 2),
hence that $\s{u}[2] \match \s{u}[1]$
(prefix $\s{u}[1..2]$ has a border of length 1)
and $\s{u}[3] \match \s{u}[2]$
(suffix $\s{u}[2..3]$ has a border of length 1),
but, since $\s{y}[3] = 0$, also that \s{u} has no border of length 1.

\cite{SW08} describes an algorithm that computes the prefix array
of any indeterminate string;
in this paper we consider the ``reverse engineering'' problem of computing a string
corresponding to a given ``feasible'' array --- that is,
any array that could conceivably be a prefix array:

\begin{defn}
\label{defn-feasible}
An integer array $\s{y} = \s{y}[1..n]$ such that $\s{y}[1] = n$ and,
for every $i \in 2..n$,
\begin{equation}
\label{condition-feasible}
0 \le \s{y}[i] \le n\+ 1\- i,
\end{equation}
is said to be \itbf{feasible}.
A feasible array that is a prefix array of a regular string is said to be \itbf{regular}.
\end{defn}
We will often use the condition
$i \le i \+ \s{y}[i] \le n\+ 1$, equivalent to (\ref{condition-feasible}).
Note that there are $n!$ distinct feasible arrays of length $n$.
Recalling that there are $(2^{\sigma}\- 1)^n$ distinct strings of length $n$
for a fixed alphabet size $\sigma$,
and applying Stirling's inequality \cite[p.\ 479]{K68}
$$n! > \sqrt{2\pi n}(n/e)^n,$$
where $e = 2.718\cdots$ is the base of the natural logarithm,
we see that (for fixed $\sigma$) the number of feasible arrays
exceeds the number of strings whenever $n$ is large enough that
\begin{equation}
\label{asymp}
\sqrt{2\pi n}\Big(\frac{n}{e(2^{\sigma}\- 1)}\Big)^n > 1.
\end{equation}

The first reverse engineering problem was introduced in
\cite{FLRS99,FGLR02}, where a linear-time algorithm was described
to compute a lexicographically least string whose border array was
a given integer array --- or to return the result that no such string exists.
There have been many such results published since;
for example, \cite{BIST03,DLL05,FS06}.
In \cite{CCR09} an $O(n)$ time algorithm is described
to solve the reverse engineering problem for a given feasible
array $\s{y} = \s{y}[1..n]$;
that is, whenever \s{y} is regular, computing a
lexicographically least regular string \s{x}
corresponding to \s{y};
and, whenever \s{y} is not regular, reporting failure.

In Section~\ref{sect-indet}, notwithstanding (\ref{asymp}),
we prove the surprising result
that every feasible array
is in fact a prefix array of some string (on some alphabet);
further, we characterize the minimum alphabet size of a regular string
corresponding to a given prefix array in terms of the largest clique
in the negative ``prefix'' graph $\mathcal{P}^-$.
We go on to give necessary and sufficient conditions
that a given prefix array is regular.
Section~\ref{sect-graphs} establishes the duality between
strings (whether regular or indeterminate) and labelled undirected graphs;
also it provides a characterization of the minimum alphabet size
of an indeterminate string \s{x} in terms of the number of
``independent'' maximal cliques in the ``associated graph" $\mathcal{G}_{\s{x}}$.
% an $\Theta(n)$-time algorithm to determine whether or not \s{y} is regular.
Section~\ref{sect-conc} outlines future work.

\section{Prefix Arrays \& Indeterminate Strings}
\label{sect-indet}

We begin with an immediate consequence of Definition~\ref{defn-prefix}:
\begin{lemm}
\label{lemm-easy}
Let $\s{x} = \s{x}[1..n]$ be a string.  An integer array $\s{y} = \s{y}[1..n]$ is the prefix array of $\s{x}$  if and only if for each position $i \in 1..n$,
the following two conditions hold:
\begin{itemize}
\item[(a)] $\s{x}\big[1..\s{y}[i]\big] \match \s{x}\big[i..i+ \s{y}[i]-1\big]$ ;
\item[(b)] if $i + \s{y}[i] \le n$, then $\s{x}\big[y[i]+1\big] \not\match \s{x}\big[i +\s{y}[i]\big]$.
\end{itemize}
\end{lemm}

We now prove the main result of this section.
\begin{lemm}
\label{lemm-feas}
Every feasible array is the prefix array of some string.
\end{lemm}
 \begin{proof}
 Consider an undirected graph $\mathcal{P} = (V, E)$ whose vertex set $V$ is the set of positions $1..n$ in a given feasible
 array $\s{y}$.  The edge set $E$ consists of the 2-element subsets $(h, k)$ such that
 \begin{equation}
\label{eq}
 h \in 1..\s{y}[i]; \ k = i+h-1
 \end{equation}
for every $i \in 2..n$.
We then define \s{x} as follows:
for each non-isolated vertex $i$,
let $\s{x}[i]$ be the set of edges incident with $i$;
for each isolated vertex $i$, let $\s{x}[i]$ be the loop $\{i, i\}$.
Let $\Sigma = E \cup L$ where $L$ is the set of loops.
We claim that $\s{y}$ is the prefix array of $\s{x} = \s{x}[1..n]$.

To see this, note that for an index $i$ such that $\s{y}[i]>0$,
Lemma \ref{lemm-easy}(a) is satisfied by construction.
Then suppose that for some $\s{y}[i] > 0$ and $i\+ \s{y}[i] \le n$,
$\s{x}\big[y[i]\+ 1\big] \match \s{x}\big[i\+ \s{y}[i]\big]$.
But this contradicts Lemma~\ref{lemm-easy}(b),
and so $$\s{x}\big[y[i]\+ 1\big] \not\match \s{x}\big[i\+ \s{y}[i]\big].$$

In case $\s{y}[i] = 0$, Lemma \ref{lemm-easy}(a) is satisfied vacuously.
Moreover, $i$ is isolated and thus $\s{x}[i] = \{i, i\}$,
which does not match $\s{x}[1]$; consequently,
Lemma  \ref{lemm-easy}(b) is again satisfied.
Therefore, $\s{y}$
 coincides with the prefix array of $\s{x}$, which is a string over the set $\Sigma'$ of subsets of $\Sigma$. \qed
 \end{proof}

\begin{figure}[t]
  \begin{minipage}{0.5\linewidth}
  \centering
  \includegraphics*[viewport = 1.5cm 12cm 9cm 19cm,scale=0.7]{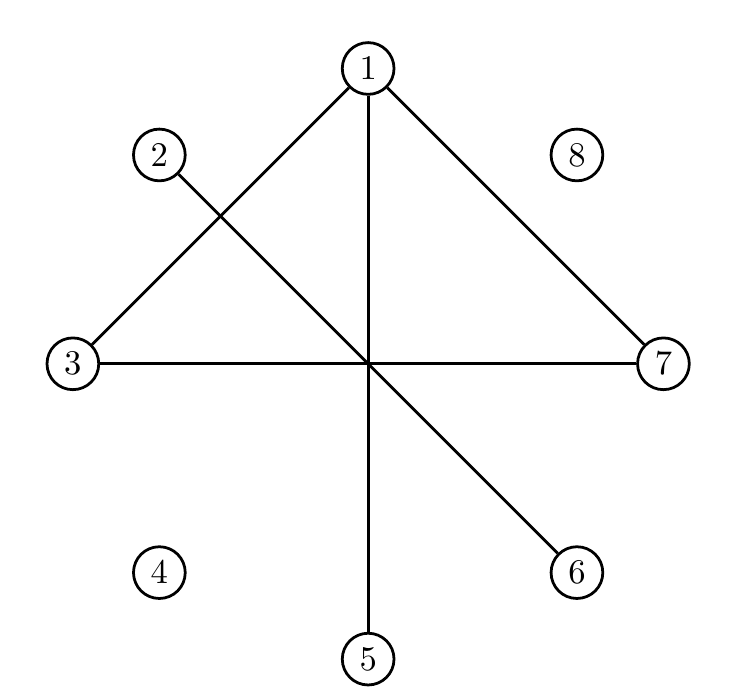}
  \caption{$\mathcal{P}_{\s{y_1}}^+$ for $\s{y_1} = 80103010$}\label{Graph-ex-1-p}
  \end{minipage}
  \hfill
  \begin{minipage}{0.5\linewidth}
  \centering
  \includegraphics*[viewport = 1.5cm 12cm 9cm 19cm,scale=0.7]{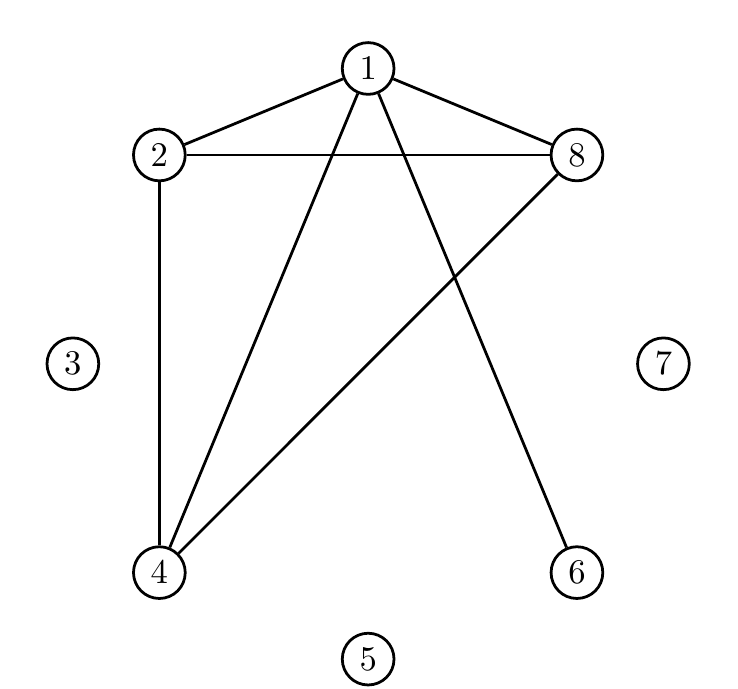}\\
  \caption{$\mathcal{P}_{\s{y_1}}^-$ for $\s{y_1} = 80103010$}\label{Graph-ex-1-n}
  \end{minipage}
\end{figure}

The construction described in this proof yields a string \s{x}
whose prefix array is \s{y}, but \s{x} is only one string among many.
For example, given
the feasible array $\s{y} = 80103010$, this construction yields (temporarily
simplifying the notation)
edges $E = \{13, 15, 26, 37, 17\}$ and loops $L = \{44,88\}$.  Relabelling these
seven edges/loops as $a,b,c,d,e,f,g$ respectively, we construct \s{x} as described in the proof of Lemma~\ref{lemm-feas}:
\begin{equation}
\label{reg}
\s{x} = \{a,b,e\}\{c\}\{a,d\}\{f\}\{b\}\{c\}\{d,e\}\{g\},
\end{equation}
an indeterminate string, when in fact \s{y} is also the prefix array
of the regular string $\s{x} = abacabad$
(and so, by Definition~\ref{defn-feasible}, itself regular).

\begin{defn}
\label{defn-G}
Let $\mathcal{P} = (V,E)$ be a labelled graph with vertex set $V = \{1,2,\ldots,n\}$
consisting of positions in a given feasible array $\s{y}.$
In $\mathcal{P}$ we define, for $i \in 2..n$, two kinds of edge
(compare Lemma~\ref{lemm-easy}):
\begin{itemize}
\item[(a)]
for every $h \in 1..\s{y}[i]$, $(h,i\+ h\- 1)$ is called a \itbf{positive edge};
\item[(b)]
$(1\+ \s{y}[i],i\+ \s{y}[i])$ is called a \itbf{negative edge}, provided $i\+ \s{y}[i] \le n.$
\end{itemize}
$E^+$ and $ E^-$ denote the sets of positive and negative edges, respectively.  We write
$E = E^+ \cup E^-$, $\mathcal{P}^+ = (V,E^+)$,
$\mathcal{P}^- = (V,E^-)$,
and we call $\mathcal{P}$ the \itbf{prefix graph} of $\s{y}$.
If $\s{x}$ is a string having $\s{y}$
as its prefix array,
then we also refer to $\mathcal{P}$ as the prefix graph
 of $\s{x}$.
\end{defn}

\begin{figure}[t]
  \begin{minipage}{0.5\linewidth}
  \centering
  \includegraphics*[viewport = 1.5cm 12cm 9cm 19cm,scale=0.7]{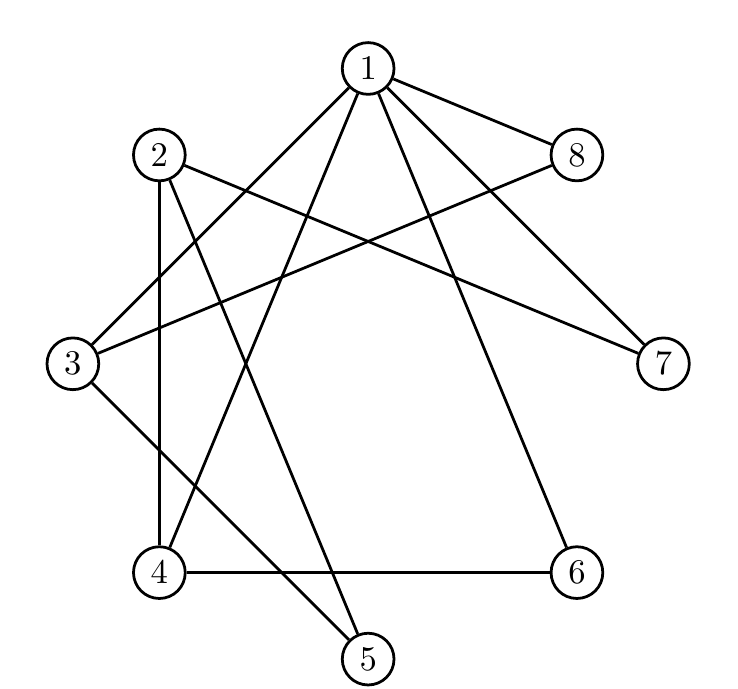}\\
  \caption{$\mathcal{P}_{\s{y_2}}^+$ for $\s{y_2} = 80420311$}\label{Graph-ex-2-p}
  \end{minipage}
  \hfill
  \begin{minipage}{0.5\linewidth}
  \centering
  \includegraphics*[viewport = 1.5cm 12cm 9cm 19cm,scale=0.7]{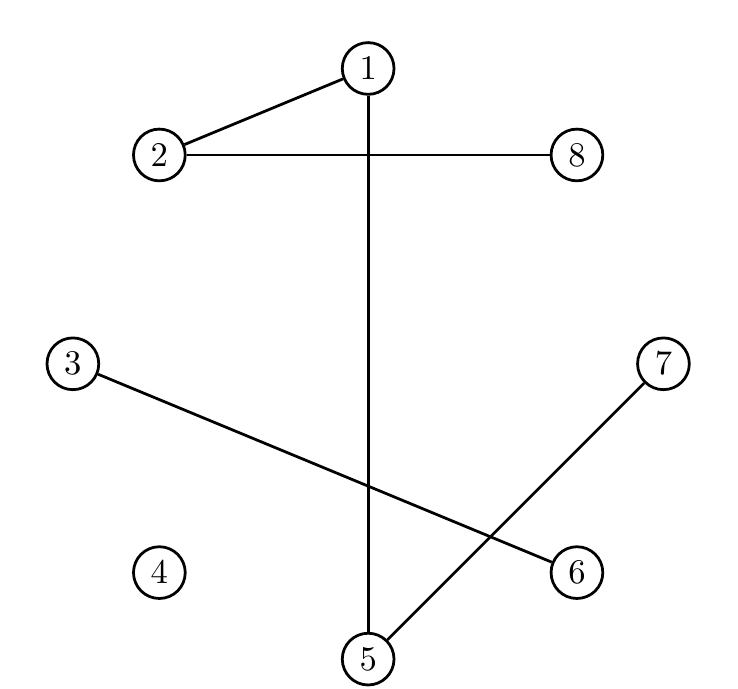}\\
  \caption{$\mathcal{P}_{\s{y_2}}^-$ for $\s{y_2} = 80420311$}\label{Graph-ex-2-n}
  \end{minipage}
\end{figure}

Figures \ref{Graph-ex-1-p}--\ref{Graph-ex-2-n} show the prefix graphs
for
$$ \begin{array}{rccccccc}
\scriptstyle 1 & \scriptstyle 2 & \scriptstyle 3 & \scriptstyle 4 & \scriptstyle 5 & \scriptstyle 6 & \scriptstyle 7 & \scriptstyle 8 \\
\s{y_1} = 8 & 0 & 1 & 0 & 3 & 0 & 1 & 0  \\
\s{y_2} = 8 & 0 & 4 & 2 & 0 & 3 & 1 & 1
\end{array} $$

From Definition~\ref{defn-G} it is clear that
\begin{remk}
\label{rem-Py}
For every feasible array \s{y}, there exists one and only one
prefix graph $\mathcal{P}$,
which therefore may be written $\mathcal{P}_{\s{y}}$;
moreover, $\mathcal{P}_{\s{y}} = \mathcal{P}_{\s{y}'}$ if and only if
$\s{y} = \s{y'}$.
\end{remk}

Recall that a graph $\mathcal{G} = (V,E)$ is said to be \itbf{connected} if every
pair of vertices in $V$ is joined by a path in $E$.
A \itbf{connected component} (or \itbf{component}, for short) of $\mathcal{G}$
is a subgraph $\mathcal{G'} = (V',E')$ formed on a largest subset $V' \subseteq V$ such
that every pair of vertices $i,j \in V'$ is joined by a path
formed from edges $E' \subseteq E$.
The graph $\mathcal{P^+}$ of Figure~\ref{Graph-ex-1-p} has two disjoint
connected components, while that of Figure~\ref{Graph-ex-2-p} has only one.

The basic properties of the prefix graph $\mathcal{P}_{\s{y}}$ of a feasible array $\s{y} = \s{y}[1..n]$ are as follows:
\begin{lemm}
\label{lemm-G}
Let $\mathcal{P} = \mathcal{P}_{\s{y}}$ be the prefix graph corresponding to
a given feasible array \s{y}.
\begin{itemize}
\item[(a)] $E^+$ and $E^-$ are disjoint and $|E^-| = n\- s$ where $s$ is the number of indices
$i \in 1..n$
for which $i\+ \s{y}[i] = n+1.$  For every $i \in 2..n$, either $(1,i)\in E^+$ or $(1,i)\in E^-$.
\item[(b)] If $(i,j)\in E^-$, where $i<j$,
then $\s{y}[j\- i\+ 1] = i\- 1$, and
for every $h \in 1..i-1$, $(h, j-i+h) \in E^+$.
\item[(c)] $\s{y}$ is regular if and only if every edge of $\mathcal{P}^-$ joins two vertices in disjoint connected components of $\mathcal{P}^+$.
\end{itemize}
\end{lemm}

 \begin{proof}
 \noindent
 \begin{itemize}
 \item[(a)]
 First fix $i$ and consider edges $(h, k)$, where $k\- h = i\- 1$.
 If $(p\+ 1, p\+ i)\in E^-$ is such an edge, then the
 edges in $E^+$ must satisfy $1 \le h \le p$ and therefore are distinct from $(p\+ 1, p\+ i)$.
 This shows that $E^+$ and $E^-$ are disjoint.
 Secondly, $|E^-| = n\- s$ since there is exactly one negative edge for each of the possible values of $i$,
except those for which $i\+ \s{y}[i] = n\+ 1$.
 Finally, it is easily seen from Definition \ref{defn-G} that $(1,i)$ is a positive
 edge if $\s{y}[i]$ is positive, whereas $(1,i)$ is a negative edge if $\s{y}[i]=0$.
 \item[(b)] The first statement follows from rewriting Definition~\ref{defn-G}(b)
with $j = i\+ \s{y}[i]$, the second directly from Definition~\ref{defn-G}(a).
\item[(c)] [if] Suppose that every negative edge
joins two vertices in disjoint connected components of $\mathcal{P}^+$.
Form a regular string $\s{x}$ as follows: for each component $C$ of $\mathcal{P}^+$, assign a unique identical letter, say $\lambda_C$, to all positions $\s{x}[i]$ for which $i \in C$.  We show that $\s{y}$ is the prefix array of $\s{x}[1..n]$ and therefore that $\s{y}$ is regular.  Fix a value $i \in 2..n$.  For any $j$ such that $1 \le j \le \s{y}[i]$, $(j, j+i-1)$ is a positive edge.
 Thus $j$ and $j+i-1$ are in the same component of $\mathcal{P}^+$, and hence $\s{x}[j] = \s{x}[j+i-1]$.  We also note
that $(\s{y}[i]+1, \s{y}[i]+i)$ is a negative edge
(provided $\s{y}[i]\+ i \le n$).
If so, then by hypothesis $\s{y}[i]+1$ and $\s{y}[i]+i$ lie
in disjoint components of $\mathcal{P}^+$, so that,
by the uniqueness of $\lambda_C$,
$\s{x}\big[\s{y}[i]+1\big] \not\match \s{x}\big[\s{y}[i]+i\big]$.
This is precisely what we need in order to conclude
that $\s{y}$ is the prefix array of $\s{x}[1..n]$.
Since \s{x} is regular, so is \s{y}, as required. \\ 

\item[ ] [only if]
Suppose that \s{y} is regular,
therefore the prefix array of a regular string $\s{x}$.
Now consider any negative edge $(p, q)$ of the prefix graph $\mathcal{P}$
of \s{y}, so that by Lemma~\ref{lemm-easy}(b)
$\s{x}[p] \not\approx \s{x}[q]$.
If $p$ and $q$ were in the same component of $\mathcal{P}^+$,
we would have by Lemma~\ref{lemm-easy}(a)
a path in $\mathcal{P}^+$ joining
$p$ to $q$ consisting of edges $(h, k)$ such that $\s{x}[h] \match \s{x}[k]$.
By the regularity of \s{y}, this requires $\s{x}[h] = \s{x}[k]$,
so that $\s{x}[p] = \s{x}[q]$, a contradiction.  \qed
 \end{itemize}
 \end{proof}

From Definition \ref{defn-G}, we see that $|E^+|$ can be as small as $0$
(for example, when $\s{x} = ab^{n-1}$)
or as large as $ {n \choose 2}$ (when $\s{x} = a^n$).
From Lemma~\ref{lemm-G}(b) we see that many of the edges in $E^+$ can be
deduced from those in $E^-$.
In fact, if we add an extra node $n \+1$ and also, in the cases $i > 1$
for which $i\+ \s{y}[i] = n\+ 1$ --- that is, whenever $\s{x}$ has a border of
length $\s{y}[i] = n\+ 1\- i$ ---, add the edges
$(1+\s{y}[i], n+1)$ to $E^-$, then all of $E^+$ can be deduced from $E^-$.
Let us call this graph with the additional node and edges
the \itbf{augmented prefix graph} and denote it by $\hat{\mathcal{P}}$
with corresponding edge sets $\hat{E}^+ = E^+$ and $\hat{E}^-$.
By Lemma~\ref{lemm-G}(a), $\hat{E}^-$ consists of exactly $n\- 1$ edges,
which together determine $O(n^2)$ edges in $E^+$.
Of course the converse is also true:
$E^+$ determines $\hat{E}^-$.
Hence, from Remark~\ref{rem-Py},
either $\mathcal{P}^+$ or $\hat{\mathcal{P}}^-$ is sufficient to determine
a corresponding prefix array \s{y}.

However, a bit more can be said.  From Lemma~\ref{lemm-G}(b)
we see that every edge $(i,j) \in E^-$ determines
the value $\s{y}[j\- i\+ 1]$ of a position $j\- i\+ 1$ in \s{y}.
Thus a simple scan of \s{y} can identify all positions $h$
that are {\it not} determined by $E^-$;
for all such $h$, it must be true that $\s{y}[h] = n\- h\+ 1$.
In other words $E^-$ determines $\hat{E}^-$.
Writing $A \equiv B$ to mean that $A$ can be computed from $B$, and {\it vice versa},
we may summarize this discussion as follows:
\begin{remk}
\label{remk-equiv}
$\s{y} \equiv \mathcal{P}_{\s{y}}^+ \equiv \hat{\mathcal{P}}_{\s{y}} \equiv \hat{\mathcal{P}}_{\s{y}}^- \equiv \mathcal{P}_{\s{y}}^-$:
the prefix array and the negative prefix graph provide
the same information and so determine
the same set of $($not necessarily regular$)$ strings \s{x}.
\end{remk}

Recall \cite[p.\ 188]{BM08} that a \itbf{$t$-clique} in a graph $\mathcal{G}$
is a complete subgraph $K_t$ of $\mathcal{G}$ on $t$ vertices,
while the \itbf{clique number} $\omega = \omega(\mathcal{G})$ is
the order $t$ of the largest clique.
We say that a $t$-clique is \itbf{maximal} if
it is not a subclique of any $(t\+ 1)$-clique.
Note that, since every isolated vertex is a complete subgraph,
$E =  \emptyset \Leftrightarrow \omega = 1$.

\begin{defn}
\label{defn-regular-G}
If $\s{y}$ is a regular feasible array,
then its prefix graph $\mathcal{P}_{\s{y}}$ is also said
to be \itbf{regular}.
\end{defn}

We use these ideas to characterize the minimum alphabet size
of any regular string with a given prefix graph $\mathcal{P}$.
% First two preliminary lemmas:
% 
% \begin{lemm}
% \label{lemm-clique}
% Suppose \s{x} is a string whose prefix graph is $\mathcal{P}$.
% Then every maximal $t$-clique in $\mathcal{P}^-$ contains a vertex $i \ge 1$
% such that $\s{x}[i] \approx \s{x}[1]$.
% \end{lemm}
% \begin{proof}
% Consider any maximal $t$-clique $C_t$ in $\mathcal{P}^-$
% and suppose that every vertex $i$ in $C_t$
% is such that $\s{x}[i] \not\approx \s{x}[1]$.
% But then $(1,i) \in E^-$ for every $i \in C$,
% which therefore must include vertex 1, a contradiction.  \qed
% \end{proof}
% \begin{lemm}
% \label{lemm-clique'}
% Suppose \s{x} is a regular string whose prefix graph is $\mathcal{P}$.
% Consider a maximal $t$-clique $C_t = \{i_1,i_2,\ldots,i_t\}$ in $\mathcal{P}^-$,
% $i_1 < i_2 < \ldots < i_t$.
% Then either $i_1 = 1$ or there exists some $j$ such that
% $\s{x}[i_j] = \s{x}[1]$, where $C_{t'}$,
% formed by replacing $i_j$ by $1$, is also a $t$-clique.
% \end{lemm}
% \begin{proof}
% By Lemma~\ref{lemm-clique} the $t$-clique $C_t$ contains vertex $i$
% such that $\s{x}[i] \approx \s{x}[1]$.
% If $i = 1$, the lemma holds; if not,
% then, since \s{x} is regular, $\s{x}[1]$ does not match
% any other vertex in $C_t$ and so $\s{x}[1]$ can replace $\s{x}[i]$,
% yielding the second $t$-clique.  \qed
% \end{proof}
% 
Consider the edges $(i,j),\ i < j$, of regular $\mathcal{P}^-$,
in ascending order of $j$.
Suppose without loss of generality that \s{x}
is defined on the alphabet $\Sigma$ of consecutive positive integers
(so that the ordering of \s{x} is with respect to $\Sigma$).
Figure~\ref{alg-greedy} describes an on-line algorithm ASSIGN
that, from the sorted list of edges in $\mathcal{P}^-$,
computes a lexicographically least string \s{x}
on $t = \omega(\mathcal{P}^-)$ letters whose
prefix graph is $\mathcal{P}$.

\begin{figure}[h!]
{\leftskip=0.5in\obeylines\sfcode`;=3000
\bproc ASSIGN $(\mathcal{P}^-,\s{x})$
Radix sort the edges $(i,j)$, $i < j$, of $\mathcal{P}^-$ by $j$.
$t \la 1;\ N[t] \la 0$
\bfor $j \la 1$ \bto $n$ \bdo
\com{Get all the edges of $\mathcal{P}^-$ with largest vertex $j$.}
\qq $S \la \{(i_1,j),(i_2,j),\ldots,(i_r,j)\}$
\qq \bif $r = 0$ \bthen $\s{x}[j] \la 1$
\com{Thus, if $\mathcal{P}^-$ has no edges, $\s{x} = 1^n$.}
\qq \belse
\qq \com{Determine the least letter $\ell$ that does {\it not} occur}
\qq \com{at any position $i_h$ in $S$; possibly $\ell = t\+ 1$.}
\qq\qq \bfor $h \la 1$ \bto $r$ \bdo $N\big[\s{x}[i_h]\big] \la 1$
\qq\qq $\ell \la 1$
\qq\qq \bwhile $\ell \le t$ \band $N[\ell] = 1$ \bdo $\ell \la \ell\+ 1$
\qq\qq \bif $\ell > t$ \bthen $t \la \ell;\ N[t] \la 0$
\qq\qq \bfor $h \la 1$ \bto $r$ \bdo $N\big[\s{x}[i_h]\big] \la 0$
\qq\qq $\s{x}[j] \la \ell$
\caption{Given the negative prefix graph $\mathcal{P}^-$ of a prefix graph $\mathcal{P}$ known to be regular, compute a lexicographically least string \s{x} on $t = \omega(\mathcal{P}^-)$ letters whose prefix graph is $\mathcal{P}$.}
\label{alg-greedy}
}
\end{figure}

Algorithm ASSIGN maintains a bit vector $N$ that,
for each $j$, specifies the letters $\s{x}[i]$ that have occurred
at positions $(i,j) \in E^-$ --- that is, $N\big[\s{x}[i]\big] = 1$.
Observe that a new letter $t\+ 1$ is added
if and only if vertex $j$ has an edge to vertices representing
{\em all} previous letters $1..t$.
This is true for every $t \ge 1$.
Thus letter $t\+ 1$ is introduced if and only if
there are already $t$ vertices that form a clique in $\mathcal{P}^-$.
Consequently the number of letters used by the algorithm to form \s{x}
is exactly $t = \omega(\mathcal{P}^-)$.
Note also that the letter assigned
at each position $j$ is least with respect to the preceding letters,
whether the letter is a new one in the string or not.
Since the letters are introduced from left to right and never changed,
\s{x} must therefore be lexicographically least with respect to $\mathcal{P}^-$.
Note further that, since position $j$ in the lexicographically least \s{x} is determined
for $j = 1,2,\ldots,n$ based solely on preceding positions $i < j$,
it suffices to use $\mathcal{P}^-$ rather than the augmented $\hat{\mathcal{P}}^-$,
in accordance with Remark~\ref{remk-equiv}.

Next consider the time requirement of Algorithm ASSIGN.
Since we know from Lemma~\ref{lemm-G}(a) that
$\mathcal{P}^-$ has at most $n\- 1$ edges,
it follows that the radix sort can be performed in $O(n)$ time.
For the same reason,
within the \bfor loop, formation of the set $S$ also
has an overall $O(n)$ time requirement.
The processing that updates the bit vector $N$,
in order to determine the least letter $\ell$ to be assigned to $\s{x}[j]$,
requires $\Theta(r)$ time, where $r$ is the size of $S$,
in order to set both $N\big[\s{x}[i_h]\big] \la 1$
and $N\big[\s{x}[i_h]\big] \la 0$;
in addition the \bwhile loop requires $O(r)$ time in the worst case.
Since $|E^-| \le n\- 1$, it follows that
the sum of all $|S| = r$ is $O(n)$,
and so the overall time requirement of this processing is $O(n)$.

\begin{lemm}
\label{lemm-alg}
For a regular prefix graph $\mathcal{P}$ on $n$ vertices, Algorithm ASSIGN
computes in $O(n)$ time a lexicographically least string
on $t = \omega(\mathcal{P}^-)$ letters whose prefix graph is $\mathcal{P}$.
\end{lemm}
\begin{proof}
We need to show that the string \s{x} computed by the algorithm
is indeed consistent with $\mathcal{P}$
(that is, by Remark~\ref{remk-equiv}, the corresponding prefix array \s{y}).
Observe that $S$ is always empty for $j = 1$,
so that therefore the initial assignment $\s{x}[1] \la 1$
is consistent with the subgraph $\mathcal{P}_1$ on a single vertex.
Suppose then that $\s{x}[1..j\- 1]$ has been computed by ASSIGN
for some $j \in 2..n$ so as to be consistent with with the
subgraph $\mathcal{P}_{j-1}$ on vertices $1,2,\ldots,j\- 1$.
For the addition of vertex (position) $j$, there are three possibilities:
\begin{description}
\item[\boldmath $|S| = 0$.]
In this case, $\s{x}[j] \la 1$,
the least letter, so that $\s{x}[j] = \s{x}[1]$,
and therefore
$\s{x}[1..j]$ remains consistent with $\mathcal{P}^-_j = \mathcal{P}^-_{j-1}$.
\item[{\bf $S$ gives rise to $t$ distinct letters.}]
Here $\s{x}[j] \la t\+ 1$, a new letter.
Since this is the first occurrence of $t\+ 1$ in \s{x},
and since there is no alternative, therefore
$\s{x}[1..j]$ is again consistent with $\mathcal{P}_j$
and has only the empty border.
\item[{\bf $S$ gives rise to $t' < t$ distinct letters.}]
From the set $S$ we know that $\s{x}[1..j\- 1]$ has exactly $r$ borders
not continued to $\s{x}[1..j]$.
The longest of these borders is $\s{x}[1..i_r\- 1]$.
There may be a border of $\s{x}[1..j\- 1]$
that is on the other hand actually continued to $\s{x}[1..j]$.
If not, then the assignment $\s{x}[j] \la \ell$ is consistent with $\mathcal{P}_j$,
where $\ell$ is the least letter not precluded by $S$.
Suppose then that there exists a border $\s{x}[1..i] = \s{x}[j\- i\+ 1..j]$,
$i \ge 1$.
Note that while there may be more than one such border, $\s{x}[i]$
must be the same for each one, since we suppose that \s{x} is regular.
Furthermore, $\s{x}[i]$ was chosen by the algorithm
to be a minimum letter $\ell_i$ with respect to the prefix $\s{x}[1..i\- 1]$;
since $\s{x}[j\- i\+ 1..j\- 1] = \s{x}[1..i\- 1]$,
the choice of a minimum letter with respect to $\s{x}[1..j\- 1]$
must yield $\ell_j = \ell_i$,
hence also consistent with $\mathcal{P}_j$.
\end{description}
Therefore by induction the lexicographically least string $\s{x}[1..j]$
is consistent with $\mathcal{P}_j$.
We have argued above that \s{x} is lexicographically least,
also that the time requirement of the algorithm is $O(n)$.
Thus the lemma is proved.  \qed
\end{proof}

Notice that the alphabet size determined by ASSIGN is least possible,
given $\mathcal{P}$.
Instead of assigning letters to positions in \s{x},
we could just as well have labelled vertices of $\mathcal{P}$
with these letters;
thus we have

\begin{cor}
\label{cor-clique}
The class of regular negative prefix graphs
$\mathcal{P}^-$ has the property that the chromatic number $($minimum alphabet size$)$
$\chi(\mathcal{P}^-) = \omega(\mathcal{P}^-)$
for every graph in the class.
\end{cor}
This property does not hold in general;
in \cite{M55}, for example,
it is shown that there exist triangle-free graphs $\mathcal{G}$
($\omega(\mathcal{G}) = 2$)
with arbitrarily large chromatic number.

To get a sense of the labelling, consider the following regular prefix array
$$ \begin{array}{rcccccccccccccccccccc}
 & \scriptstyle 1 & \scriptstyle 2 & \scriptstyle 3 & \scriptstyle 4 & \scriptstyle 5 & \scriptstyle 6 & \scriptstyle 7 & \scriptstyle 8 & \scriptstyle 9 & \scriptstyle 10 & \scriptstyle 11  & \scriptstyle 12 & \scriptstyle 13 & \scriptstyle 14 & \scriptstyle 15 & \scriptstyle 16 & \scriptstyle 17 & \scriptstyle 18 & \scriptstyle 19 & \scriptstyle 20 \\ 
\s{y} = & 20 & 0 & 1 & 0 & 3 & 0 & 3 & 0 & 3 & 0 & 1 & 0 & 7 & 0 & 1 & 0 & 4 & 0 & 1 & 0
\end{array}$$
whose corresponding $\mathcal{P}^-_{\s{y}}$ has edges
(sorted as in Algorithm ASSIGN)
\begin{eqnarray*}
& & (1,2),\ (1,4),\ (2,4),\ (1,6),\ (1,8),\ (4,8),\ (1,10),\ (4,10), \\ 
& & (1,12),\ (2,12),\ (4,12),\ (1,14),\ (1,16),\ (2,16),\ (1,18), \\ 
& & (1,20),\ (2,20),\ (8,20).
\end{eqnarray*}
$\mathcal{P}^-_{\s{y}}$ has a single maximal clique, $(1,2,4,12)$, on four vertices,
and the corresponding
lexicographically least string is
$$\s{y} = abacabababadabacabac.$$
Note that $\hat{\mathcal{P}}^-_{\s{y}}$ contains in addition
the edge $(5,21)$ not required for the lexicographically least \s{x}.

Now consider $t$-cliques $\{i_1,i_2,\ldots,i_t\}$ (not necessarily maximal)
in regular prefix arrays $\mathcal{P}^-$ for which $i_1 = 1$,
together with regular strings \s{x} whose prefix graph is $\mathcal{P}$.
A 1-clique corresponds to a prefix $\s{p_1} = \lambda_1$ of \s{x},
where $\lambda_1$ is some (say, smallest) letter.
Then for every 2-clique $(1,i_2)$ in $\mathcal{P}^-$,
there must exist a corresponding prefix
\s{p_2} of \s{x} such that
$$\s{p_2} = \lambda_1\s{w_1}\lambda_2,$$
where $\lambda_2 > \lambda_1$.
Similarly, for every 3-clique $(1,i_2,i_3)$
in $\mathcal{P}^-$,
there exists a corresponding prefix \s{p_3} of \s{x} such that
\begin{eqnarray*}
\s{p_3} & = & \lambda_1\s{w_1}\lambda_2\s{w_2}\lambda_1\s{w_1}\lambda_3 \\ 
& = & \s{p_2}\s{w_2}\s{p'_2},
\end{eqnarray*}
where $\s{p_2},\s{p'_2}$ are identical but for
distinct rightmost letters $\lambda_2$ and $\lambda_3 > \lambda_2$, respectively.
In general, for every $t$-clique
$(1,i_2,i_3,\ldots,i_t)$ in $\mathcal{P}^-$,
there exists a corresponding prefix \s{p_t} of \s{x} such that
$$\s{p_t} = \s{p_{t-1}}\s{w_{t-1}}\s{p'_{t-1}},$$
where $\s{p_{t-1}},\s{p'_{t-1}}$ are prefixes identical
but for rightmost letters $\lambda_{t-1}$ and $\lambda_t > \lambda_{t-1}$, respectively.
Thus every $t$-clique in regular $\mathcal{P}^-$ corresponds to
a prefix of the corresponding string \s{x} that has $t\- 1$ borders
of lengths $1,2,\ldots,t\-1$.
The length of this prefix can be minimized by choosing
every $\s{w_j}$, $j \in 1..t\- 1$, to be empty, so that
the strings $\s{p_j}$ double in length at each step:
hence there exists a prefix graph on $2^{t-1}$ vertices
(or, equivalently, a feasible array of length $2^{t-1}$)
whose corresponding strings cannot be implemented on less than $t$ letters.
Thus we are able to verify a result given in \cite[Proposition 8]{CCR09}:

\begin{lemm}
\label{lemm-log}
For a given regular feasible array $\s{y} = \s{y}[1..n]$, a regular string \s{x} whose prefix array is $\s{y}$ can be constructed using at most $\floor{\log_2 n}\+ 1$ letters.
\end{lemm}

\cite{CCR09} describes a lemma more complex than Algorithm ASSIGN,
but that does not require a regular prefix array as input:
a nonregular feasible array is rejected at the first position detected.

We conclude this section with two equivalent
necessary and sufficient conditions for \s{y} to be regular.
A string \s{x} is said to be \itbf{strongly indeterminate}
(INDET, for short) if and only if its prefix array is not regular.
Recall from Definition~\ref{defn-feasible} that a feasible array is regular
if and only if it is a prefix array of a regular string.
Thus, for example, the string (\ref{reg}),
although certainly indeterminate, is not INDET because it is consistent
with the feasible array $\s{y} = 80103010$ that is a prefix array of the
regular string $\s{x} = abacabad$.
If on the other hand \s{y} is not regular, then as we have seen
(Lemma~\ref{lemm-G}(c)) there must exist
a position $i$ such that $\s{x}[i] \match \s{x}[r]$ and
$\s{x}[i] \match \s{x}[s]$, while $\s{x}[r] \not\match \s{x}[s]$,
for some positions $r$ and $s$;
in such a case we say that $\s{x}[i]$ is \itbf{INDET}.
(In terms of the prefix graph $\mathcal{P}$,
$(i,r) \in E^+,\ (i,s) \in E^+,\ (r,s) \in E^-$.)

We state two versions of what is essentially the same lemma;
we prove the second.

\begin{lemm}
\label{lemm-indet}
Suppose that $\s{x} = \s{x}[1..n]$ is a nonempty string
with prefix array $\s{y}$.
Then for $i \in 1..n$, $\s{x}[i]$ is INDET (and so therefore also \s{x}) if and only if
there exist positions $r$ and $s > r$ such that
$\s{y}[s\- r\+ 1] = r\- 1$
and one of the following holds:
\begin{itemize}
\item[(a)]
$\s{y}[r\- i\+ 1] \ge i,\ \s{y}[s\- i\+ 1] \ge i$ $(1 \le i < r < s \le n)$;
\item[(b)]
$\s{y}[i\- r\+ 1] \ge r,\ \s{y}[s\- i\+ 1] \ge i$ $(1 \le r < i < s \le n)$;
\item[(c)]
$\s{y}[i\- r\+ 1] \ge r,\ \s{y}[i\- s\+ 1] \ge s$ $(1 \le r < s < i \le n)$.
\end{itemize}
\end{lemm}

\begin{figure}[t]
  %\begin{minipage}{0.55\linewidth}
  \centering
  \raisebox{.6cm}{(a) \ }
  \includegraphics[viewport = 5cm 23cm 16cm 25cm,scale=0.9]{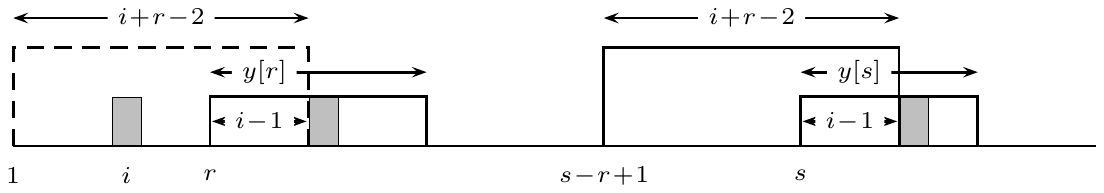}\\
  %\caption{$\mathcal{G}^+$ for $\s{y_1} = 80103010$}\label{fig-lem-16-a}
  %\end{minipage}
  %\begin{minipage}{0.55\linewidth}
  %\centering
  \raisebox{.6cm}{(b) \ }
  \includegraphics[viewport = 5cm 23cm 14cm 25cm,scale=0.9]{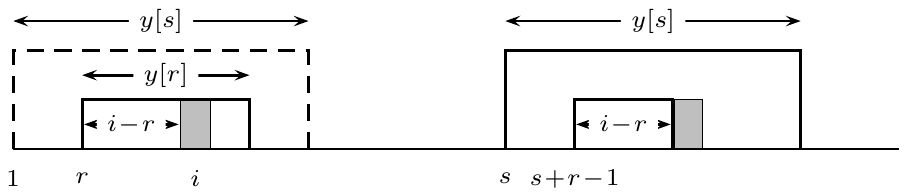}\\
  %\caption{$\mathcal{G}^-$ for $\s{y_1} = 80103010$}\label{fig-lem-16-b}
  %\end{minipage}
  %\hfill
  %\begin{minipage}{0.55\linewidth}
  %\centering
  \raisebox{.6cm}{(c) \ }
  \includegraphics[viewport = 5cm 23cm 14cm 25cm,scale=0.9]{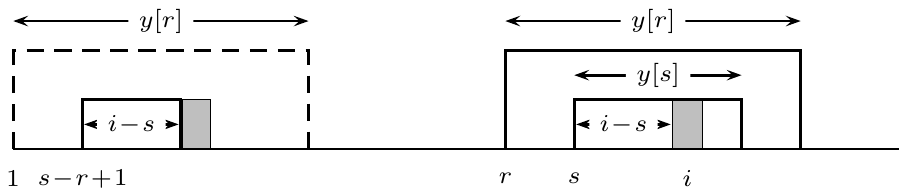}\\
  %\caption{$\mathcal{G}^-$ for $\s{y_1} = 80103010$}\label{fig-lem-16-c}
  %\end{minipage}
  \caption{The three cases of Lemma~\ref{lemm-indeta}.}
  \label{fig-3cases}
\end{figure}

\begin{lemm}
\label{lemm-indeta}
Suppose that $\s{x} = \s{x}[1..n]$ is a nonempty string
with prefix array $\s{y}$.
Then for $i \in 1..n$, $\s{x}[i]$ is INDET (and so therefore also \s{x}) if and only if
there exist positions $r$ and $s$ such that
one of the following holds:
\begin{itemize}
\item[(a)]
$\s{y}[r] \ge i,\ \s{y}[s] \ge i,\ \s{y}[s\- r\+ 1] = i\+ r\- 2$;
\item[(b)]
$r\+ \s{y}[r] > i,\ \s{y}[s] \ge i,\ \s{y}[s\+ r\- 1] = i\- r$;
\item[(c)]
$r\+ \s{y}[r] > i,\ s\+ \s{y}[s] > i,\ \s{y}[s\- r\+ 1] = i\- s$.
\end{itemize}
\end{lemm}
\begin{proof}
If $\s{x}[i]$ is INDET, then there must exist positions $r'$ and $s'$ such that
$\s{x}[i] \match \s{x}[r'],\ \s{x}[i] \match \s{x}[s'],\ \s{x}[r'] \not\match \s{x}[s']$.
Conversely, if such $r'$ and $s'$ exist, then $\s{x}[i]$ is INDET.
Without loss of generality, suppose that $s' > r'$.
Then three cases arise depending on the relative values of the distinct integers $i,r',s'$
(see Figure~\ref{fig-3cases}):
\begin{itemize}
\item[(a)]
$(1 \le i < r' < s' \le n)$
Since $\s{x}[i] \match \s{x}[r']$ and $i < r'$,
it follows that $\s{x}[1..i] \match \s{x}[r'\- i\+ 1..r']$,
hence that $\s{y}[r'\- i\+ 1] \ge i$;
similarly, $\s{y}[s'\- i\+ 1] \ge i$.
Since $\s{x}[r'] \not\match \s{x}[s']$ and $r' < s'$,
therefore $\s{y}[s'\- r'\+ 1] = r'\- 1$.
Setting $r \la r'\- i\+ 1,\ s \la s'\- i\+ 1$ yields the desired result.
\item[(b)]
$(1 \le r' < i < s' \le n)$
Since $\s{x}[i] \match \s{x}[r']$ and $r' < i$,
therefore $\s{x}[1..r'] \match \s{x}[i\- r'\+ 1..i]$,
and so $\s{y}[i\- r'\+ 1] \ge r'$;
as in (a), $\s{y}[s'\- i\+ 1] \ge i$.
Also as in (a), $\s{y}[s'\- r'\+ 1] = r'\- 1$.
Setting $r \la i\- r'\+ 1,\ s \la s'\- i\+ 1$ yields the result.
\item[(c)]
$(1 \le r' < s' < i \le n)$
As in (b), $\s{y}[i\- r'\+ 1] \ge r'$;
similarly, $\s{y}[i\- s'\+ 1] \ge s'$.
As in (a) and (b), $\s{y}[s'\- r'\+ 1] = r'\- 1$.
Setting $s \la i\- r'\+ 1,\ r \la i\- s'\+ 1$ yields the result.  \qed
\end{itemize}
\end{proof}

\section{Graphs \& Indeterminate Strings}
\label{sect-graphs}
Here we extend the ideas of Section~\ref{sect-indet}
to establish a remarkable connection between labelled graphs
and indeterminate strings.
Recall that a graph is \itbf{simple} if and only if it is
undirected and contains neither loops nor multiple edges.

We define the \itbf{associated graph},
$\mathcal{G}_{\s{x}} = (V_{\s{x}},E_{\s{x}})$, of a string \s{x} to be the simple graph whose vertices are positions $1,2,\ldots, n$  in $\s{x}$ and whose edges are the pairs $(i, j)$ such that
$\s{x}[i] \approx \s{x}[j]$.  Thus $E_{\s{x}}$ identifies {\em all} the matching
positions in \s{x}, not only those determined by the prefix array.  On the other hand, we may think of each pair
$(i,j) \not\in E_{\s{x}}$ as a \itbf{negative} edge, $\s{x}[i] \not\approx \s{x}[j]$.  Thus $\mathcal{G_{\s{x}}}$ determines all the pairs of positions in \s{x} that match or do not match each other.

It should be noted here that while $\mathcal{G}_{\s{x}}$
determines the matchings of positions in \s{x},
it does not uniquely determine the alphabet of \s{x}.
For example,
$$E_{\s{x}} = \big\{(1,2),(1,3),(1,4),(1,5),(2,3),(2,4),(2,6),(3,5),(3,6)\big\}$$
describes
$$ \begin{array}{rcccccc}
& \scriptstyle 1 & \scriptstyle 2 & \scriptstyle 3 & \scriptstyle 4 & \scriptstyle 5 & \scriptstyle 6 \\
\s{x_1} = & \{a,b,c\} & \{a,b,d\} & \{a,c,d\} & b & c & d
\end{array}$$
as well as
$$ \begin{array}{rcccccc}
& \scriptstyle 1 & \scriptstyle 2 & \scriptstyle 3 & \scriptstyle 4 & \scriptstyle 5 & \scriptstyle 6 \\
\s{x_2} = & \{a,b\} & \{a,c\} & \{b,c\} & a & b & c
\end{array}$$

Thus a given simple graph $\mathcal{G} = (V, E)$ with $n$ vertices can be the associated graph of distinct strings.  Another way to generate additional strings is by permuting the vertex labels.  Given any {\it un}labelled $\mathcal{G}$,
we can generate strings $\s{x} = \s{x}[1..n]$
by labelling the $n$ vertices $V$ of $\mathcal{G}$
with integers $1..n$, and forming a string $\s{x}$ of which $\mathcal{G}$, with this labelling, is the associated graph.  Thus an unlabelled graph $\mathcal{G}$
corresponds to a set of strings \s{x} determined by the $n!$
possible labellings of $V$.
For instance, given the graph
\begin{center}
\begin{picture}(100,20)(0,0)
\put(25,10){\circle{10}}\put(30,10){\line(1,0){30}}\put(65,10){\circle{10}}
\put(85,10){\circle{10}}
\end{picture}
\end{center}
there are six possible labellings, three of which, for example
\begin{center}
\begin{picture}(280,20)(0,0)
\put(0,10){\circle{10}}\put(5,10){\line(1,0){30}}\put(40,10){\circle{10}}
\put(60,10){\circle{10}}
\put(-2.5,7){1}\put(37.5,7){2}\put(57.5,7){3}
\put(100,10){\circle{10}}\put(105,10){\line(1,0){30}}\put(140,10){\circle{10}}
\put(160,10){\circle{10}}
\put(97.5,7){2}\put(137.5,7){3}\put(157.5,7){1}
\put(200,10){\circle{10}}\put(205,10){\line(1,0){30}}\put(240,10){\circle{10}}
\put(260,10){\circle{10}}
\put(197.5,7){3}\put(237.5,7){1}\put(257.5,7){2}
\end{picture}
\end{center}
can be chosen to lead to distinguishable regular strings
$\s{x_1} = aab,\ \s{x_2} = abb,\ \s{x_3} = aba$, respectively.
In this case the other three labellings determine the same three strings.

Consider a given string \s{x}.
Suppose that for some position $i_0 \in 1..n$,
$\s{x}[i_0]$ matches $\s{x}[i_1], \s{x}[i_2], \ldots, \s{x}[i_k]$
for some $k \ge 0$, and matches no other elements of \s{x}.
We say that position $i_0$ is \itbf{essentially regular}
if and only if the entries in positions $i_1,i_2,\ldots,i_k$
match each other pairwise.
If every position in \s{x} is essentially regular,
we say that \s{x} itself is \itbf{essentially regular}.
For example, it is easy to verify that
$$\{a,b\}\{c,d\}\{a,b\}\{e,f\}ac\{a,h\}g,$$ though indeterminate,
is essentially regular with prefix array $\s{y} = 80103010$.
On the other hand, string (\ref{reg}),
$$\s{x} = \{a,b,e\}\{c\}\{a,d\}\{f\}\{b\}\{c\}\{d,e\}\{g\},$$
also with prefix array \s{y},
is not essentially regular.  We have
\begin{lemm}
\label{lemm-union}
A string \s{x} is essentially regular if and only if
the associated graph $\mathcal{G}_{\s{x}}$ of \s{x} is
a disjoint union of cliques.
\end{lemm}
Thus combinatorics on (regular, essentially regular) words
is the study of labelled
collections of cliques.
For example, for $\s{x} = a^n$, the associated graph $\mathcal{G}_{\s{x}}$
is simply the complete graph $K_n$;
while for \s{x} such that $\s{x}[i] \match \s{x}[j]\ \Rightarrow i = j$,
$\mathcal{G}_{\s{x}}$ is $n$ copies of $K_1$.
More generally, for essentially regular \s{x},
the number of disjoint cliques in $\mathcal{G}_{\s{x}}$
is just the number of distinct letters in a regular string having 
the same associated graph as \s{x},
and the order of each clique is the number of times
the corresponding letter occurs.
% Recall from Lemma~\ref{lemm-regclique} that regular \s{x} can in fact be
% constructed using $\omega(\mathcal{P}^-)$ letters;
% thus for regular \s{x} the number of cliques in the associated graph $\mathcal{G}_{\s{x}}$
% is just $\omega(\mathcal{P}^-)$,
% the order of the maximum clique in the negative prefix graph $\mathcal{P}^-$.

Recall that
a \itbf{maximal clique} (sometimes abbreviated MC)
$K_t$ in a graph $\mathcal{G} = (V,E)$
is a clique that is not a subgraph of any other clique in $\mathcal{G}$.
% by Manolis
%\marginpar{Do we still need this par?}
Thus if $K_t$ is maximal, then for every vertex $j$ not in $K_t$,
there exists some vertex $i$ of $K_t$ such that
$(i,j) \not\in E$.
Note that every vertex of $\mathcal{G}$ must belong to at least
one maximal clique.

\begin{defn}
\label{defn-UID}
Let $\mathcal{G} = (V,E)$ be a finite simple graph,
let $S$ be the set of all MC in $\mathcal{G}$,
and let $\mathcal{I}$ be a smallest subset of $\mathcal{S}$
such that every edge of $E$ occurs at least once in $\mathcal{I}$.
Then the MC in $\mathcal{I}$ are said to be \itbf{independent} (I),
those in $\mathcal{D} = \mathcal{S}\- \mathcal{I}$ \itbf{dependent} (D).
\end{defn}

We say that an edge of $\mathcal{G}$ is a \itbf{free edge} if it belongs to exactly one MC.
Then every MC that contains a free edge is independent.

\begin{figure}[b!]
  \begin{minipage}{0.5\linewidth}
  \centering
  \includegraphics*[viewport = 1.5cm 12cm 9cm 19cm,scale=0.7]{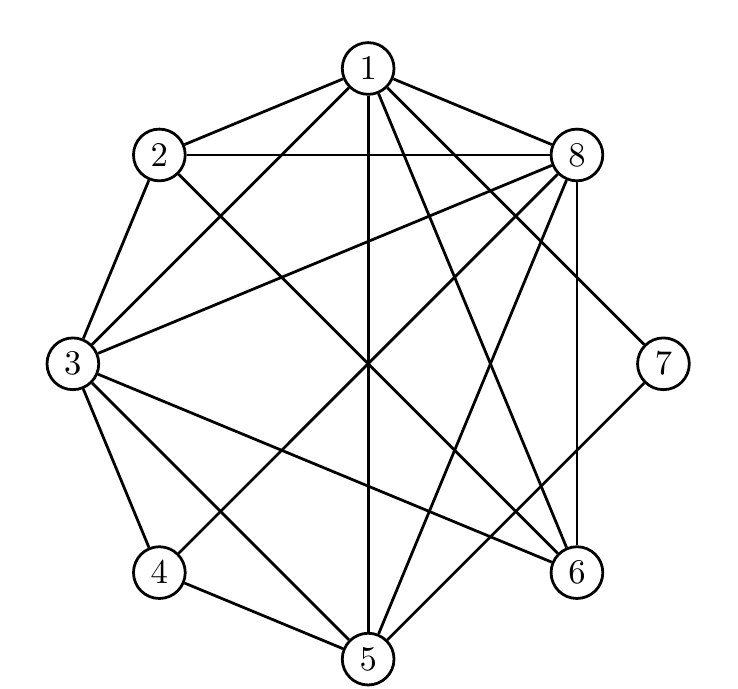}\\
  \caption[]{$\mathcal{G}_{_{\s{x}}}$ for\\ $\s{x} = \{a,b\}a\{a,c\}c\{b,c\}ab\{a,c\}$}
  \label{Graph-ex-gx-1}
  \end{minipage}
  \hfill
  \begin{minipage}{0.5\linewidth}
  \centering
  \includegraphics*[viewport = 1.5cm 12cm 9cm 19cm,scale=0.7]{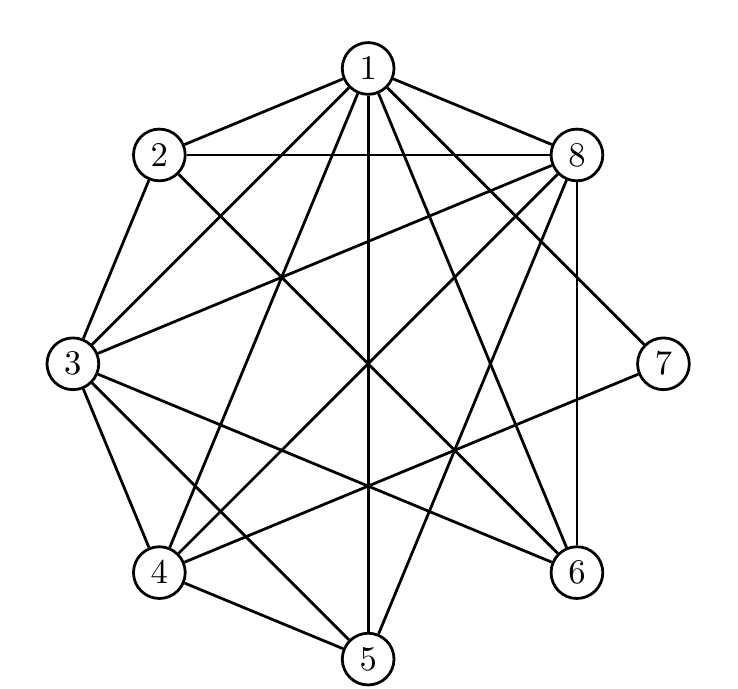}\\
  \caption[]{$\mathcal{G}_{\s{x'}}$ for\\ $\s{x'} = \{a,c,d\}a\{a,b,c\}\{b,d\}\{b,c\}ad\{a,b,c\}$}
  \label{Graph-ex-gx-2} 
  \end{minipage}
\end{figure}

We will see that for the associated graph $\mathcal{G} = \mathcal{G}_{\s{x}}$
of a string \s{x}, the independent MC are closely related to alphabet size.
Consider for example
\begin{equation}
\label{3lett}
\s{x} = \{a,b\}a\{a,c\}c\{b,c\}ab\{a,c\}.
\end{equation}
$\mathcal{G}_{\s{x}}$ (see Figure~\ref{Graph-ex-gx-1}) has four MC
\begin{equation}
\label{maxcl1}
C_1 = 12368,\ C_2 = 3458,\ C_3 = 1358,\ C_4 = 157,
\end{equation}
of which, by Definition~\ref{defn-UID},
$C_1,C_2,C_4$ are independent, since each contains at least
one free edge ($(1,2),(3,4),(1,7)$, respectively).
However, $1358$ is dependent, since its adjacencies
all occur elsewhere
($138$ is a subclique of $C_1$, $358$ a subclique of $C_2$,
$15$ an edge of $C_4$,
and so every edge of $1358$ occurs in at least one of the other three cliques).
Thus exactly three of the MC are independent,
and we see that (\ref{3lett}) has a minimum alphabet of three letters.
On the other hand, if $\mathcal{G}_{\s{x'}}$ (see Figure~\ref{Graph-ex-gx-2}) has MC
\begin{equation}
\label{maxcl2}
C_1 = 12368,\ C_2 = 3458,\ C_3 = 1358,\ C_4 = 147,
\end{equation}
all four of them are independent
(in $C_3$ the edge $15$ no longer occurs elsewhere),
and we claim that no corresponding string
\s{x'} can be constructed on fewer than four letters,
while
$$\s{x'} = \{a,c,d\}a\{a,b,c\}\{b,d\}\{b,c\}ad\{a,b,c\}$$
achieves the lower bound.

\begin{lemm}
\label{lemm-minsig}
Suppose that a graph $\mathcal{G}$ has exactly $\sigma$
independent maximal cliques.
Then there exists a string \s{x} on a base alphabet of size $\sigma$
whose associated graph $\mathcal{G}_{\s{x}} = \mathcal{G}$,
and on no smaller alphabet.
\end{lemm}
\begin{proof}
Let $\mathcal{I} = \{I_1,I_2,\ldots,I_{\sigma}\}$ be the set
of independent MC.
Suppose that initially every $\s{x}[i],\ i = 1,2,\ldots, n$, is empty;
then for $s = 1,2,\ldots,\sigma$,
form
$$\s{x}[i] \la \s{x}[i] \cup \lambda_s$$
if and only if vertex $i$ occurs in $I_s$,
where $\lambda_s$ is a unique regular letter associated
with $I_s$.
This ensures that $\s{x}[i_1] \approx \s{x}[i_2]$ if and only if
$(i_1,i_2)$ is an edge in one of the independent MC of $\mathcal{G}$.
Since by Definition~\ref{defn-UID} this includes all the edges,
it follows that $\mathcal{G} = \mathcal{G}_{\s{x}}$ is the associated graph of \s{x},
a string on a base alphabet of size $\sigma$.
Suppose that there exists a string \s{x'} on a base alphabet
of size $\sigma' < \sigma$,
where $\mathcal{G}_{\s{x'}} = \mathcal{G}_{\s{x}}$.
But then, since the regular letters in \s{x'} collectively determine
all the edges and exactly $\sigma'$ independent MC in $\mathcal{G}_{\s{x'}}$,
this means that there exists a set of independent MC in $\mathcal{G}$
of cardinality $\sigma' < \sigma$, contradicting the condition of
Definition~\ref{defn-UID} that $\mathcal{I}$ is the smallest such subset.
This completes the proof.  \qed
\end{proof}
Lemma~\ref{lemm-minsig} has an easy corollary:

\begin{lemm}
\label{lemm-easier}
Suppose that $\mathcal{G}_{\s{x}} = (V,E)$ is the associated graph
of a string \s{x} with $\sigma$ independent maximal cliques
$\mathcal{I} = \{I_1,I_2,\ldots,I_{\sigma}\}$.
\begin{itemize}
\item[(a)]
If a vertex $i \in V$ belongs to exactly $s \in 1..\sigma$
of the maximal cliques in $\mathcal{I}$,
then $|\s{x}[i]| \ge s$.
\item[(b)]
If an edge $(i,j) \in E$ belongs to exactly $s \in 1..\sigma$
of the maximal cliques in $\mathcal{I}$,
then $|\s{x}[i] \cap \s{x}[j]| \ge s$.
\end{itemize}
\end{lemm}
The following simple algorithm might be a candidate to compute a set of
independent maximal cliques:
\begin{itemize}
\item[1.]
Label I every MC that has a free edge;
\item[2.]
Alternate steps (a) and (b) until no new labellings occur:
\begin{itemize}
\item[$(a)$]
Label D each unlabelled MC with at least one edge
in an MC labelled I;
\item[$(b)$]
Label I each unlabelled MC with at least one edge in an MC
labelled D.
\end{itemize}
\end{itemize}
However, suppose that some subgraph $\mathcal{H}$ of $\mathcal{G}$
remains unlabelled after the termination of step 2 of the algorithm.
Then every edge $e$ of $\mathcal{H}$ must belong to at least two MC of $\mathcal{H}$,
since otherwise it would have been labelled in step 1.
Moreover, any MC containing $e$ cannot be labelled either I or D,
and so $\mathcal{H}$ can only be a subgraph sharing no edges
with the rest of $\mathcal{G}$ and also containing no free edges.

\begin{figure}[t]
  \begin{minipage}{0.45\linewidth}
  \centering
  \includegraphics*[viewport = 1.5cm 12cm 9cm 19cm,scale=0.7]{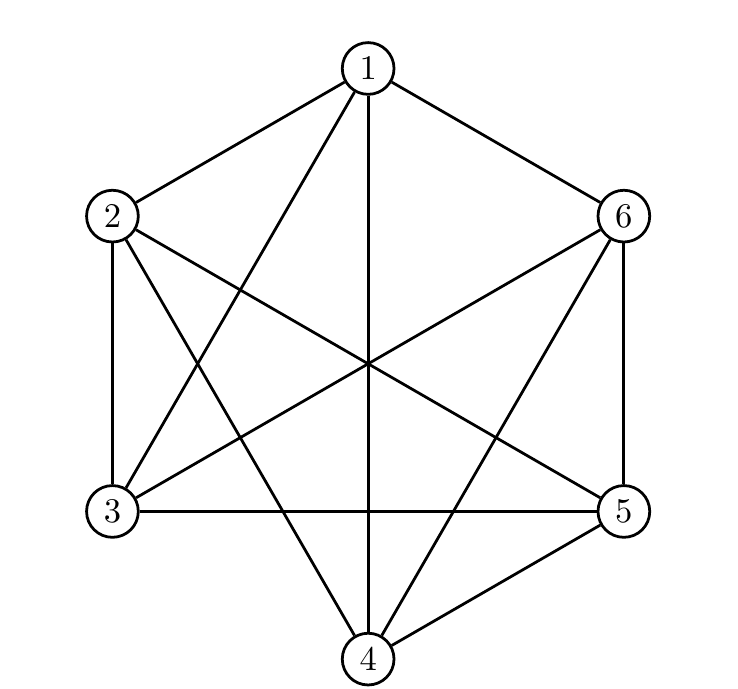}\\
\caption{Graph $\mathcal{G}$ on six vertices with eight MC,
four of them independent, and no free edges.}
  \label{Graph-ex-gx-4}
  \end{minipage}
  \hfill
  \begin{minipage}{0.45\linewidth}
  \centering
  \includegraphics*[viewport = 1.5cm 12cm 9cm 19cm,scale=0.7]{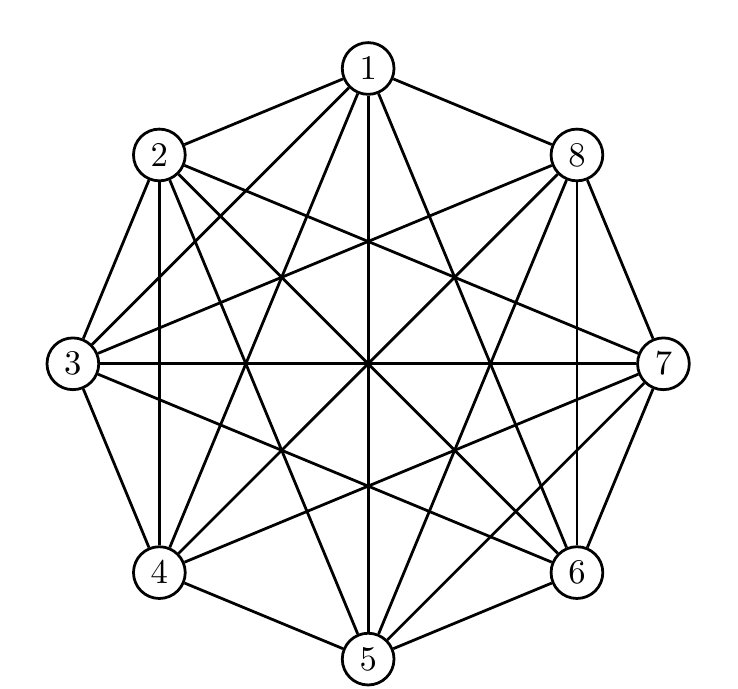}\\
\caption{Graph $\mathcal{G}$ on eight vertices with 16 MC, six of them independent,
and no free edges.}
  \label{Graph-ex-gx-3} 
  \end{minipage}
\end{figure}

To show that such a subgraph can exist,
consider the triangulated graph $\mathcal{G}$
on six vertices $V = \{1,2,3,4,5,6\}$, where
the only pairs $(i,j)$ that are {\it not} edges are
$(1,5),\ (2,6)$ and $(3,4)$, as shown in Figure~\ref{Graph-ex-gx-4}.
There are eight MC
$$123,146,245,356;\ 456,124,235,136$$
of which either the first four or the last four can be chosen to be independent,
thus by Lemma~\ref{lemm-minsig} yielding a corresponding string \s{x}
on four regular letters.
Note that every edge occurs in exactly two MC,
so that by Lemma~\ref{lemm-easier}(b) every position in
the corresponding string \s{x} contains at least two regular letters;
for example,
$$\s{x} = \{a,b\},\{a,c\},\{a,d\},\{b,c\},\{c,d\},\{b,d\}.$$

A more complex example is the graph $\mathcal{G}$ on vertices
$V = \{1,2,3,4,5,6,7,8\}$ with maximal cliques $\{1,2,3,4\}$,
$\{5,6,7,8\}$, and 14 others,
as shown in Figure~\ref{Graph-ex-gx-3}.
The only pairs $(i,j)$
that are {\it not} edges are
$(1,7),\ (2,8),\ (3,5),$ and $(4,6)$.
In this case it turns out that there are six independent MC, for example
$$1234,5678,1368,1458,2367,2457,$$
and so by Lemma~\ref{lemm-minsig} a corresponding string \s{x}
can be constructed using six regular letters
(one letter per MC):
$$\s{x} = \{a,c,d\},\{a,e,f\},\{a,c,e\},\{a,d,f\},\{b,d,f\},\{b,c,e\},\{b,e,f\},\{b,c,d\}.$$

\begin{figure}[t]
  \centering
  \includegraphics*[viewport = 1.5cm 12cm 9cm 19cm,scale=0.7]{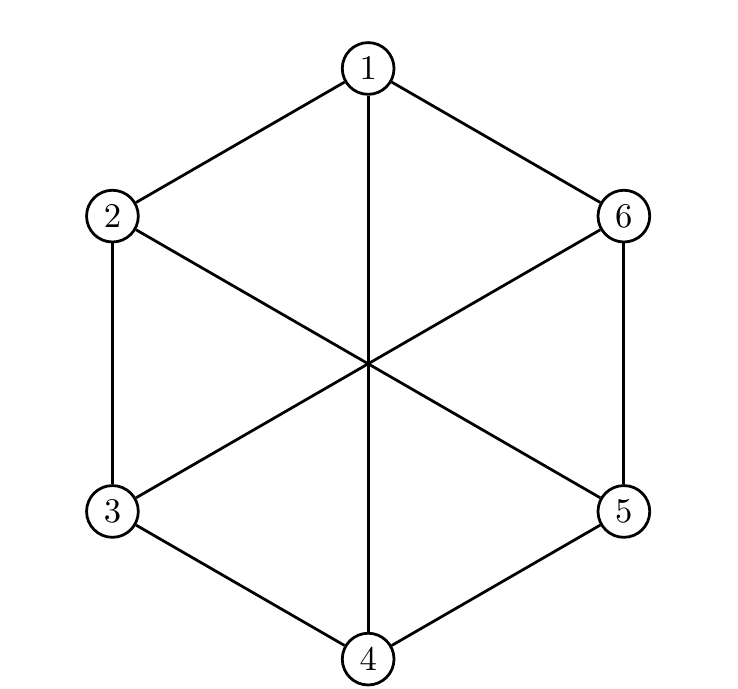}\\
  \caption{Identifying the minimum alphabet size from the number of independent maximal cliques (Lemma~\ref{lemm-minsig})}
  \label{Graph-ex-gx-5}
\end{figure}

These examples show that whenever graphs or subgraphs without free edges exist,
the identification of independent MC becomes more difficult.
In such cases we know of no algorithm to compute them apart from exhaustive search.
Thus, while it is straightforward, given \s{x},
to determine $\mathcal{G}_{\s{x}}$,
it is nontrivial, given $\mathcal{G}$, to determine
a string \s{x} on a smallest alphabet such that $\mathcal{G} = \mathcal{G}_{\s{x}}$.

From Lemma~\ref{lemm-union} it follows that the
maximum alphabet size required for an essentially regular string \s{x}
is $n$; thus to compute \s{x} from a feasible array \s{y}
is potentially an $O(n)$ algorithm and, as shown in \cite{CCR09},
is actually $O(n)$.
% by Manolis
However, for indeterminate strings, Lemma~\ref{lemm-minsig}
shows that the minimum alphabet size is
the number $\sigma$ of independent maximal cliques in
$\mathcal{G}_{\s{x}}$.
A classical result from graph theory \cite{MM65} shows that
the number of maximal cliques may be as much as $3^{n/3}$,
and so an indeterminate string potentially could require an alphabet
of exponential size.
For example, for $n = 6$, consider the graph $\mathcal{G}_{\s{x}}$
on six vertices $V_{\s{x}} = \{1,2,\ldots,6\}$
with nine edges ($9 = 3^{6/3}$)
$$E_{\s{x}} = \{(1,2),\ (1,4),\ (1,6),\ (2,3),\ (2,5),\ (3,4),\ (3,6),\ (4,5),\ (5,6)\},$$
as shown in Figure~\ref{Graph-ex-gx-5}. Each of these edges is a maximal independent 2-clique,
and so by Lemma~\ref{lemm-minsig} a corresponding string is
$$\s{x} = \{a,b,c\}\{a,d,e\}\{d,f,g\}\{b,f,h\}\{e,h,i\}\{c,g,i\},$$
defined on an alphabet of nine regular letters
with prefix array $\s{y} = 650301$.

Note here that information is lost in the transformation
from \s{x} to \s{y}.
The prefix graph $\mathcal{P}^+$ corresponding to $650301$ has the same
nine edges $E_{\s{x}}$, but $\mathcal{P}^-$ contains,
instead of the six negative edges
$$(1,3),\ (1,5),\ (2,4),\ (2,6),\ (3,5),\ (4,6)$$
implied by $E_{\s{x}}$,
just two: $E^- = \{(1,3),\ (1,5)\}$.
Thus by reverse engineering \s{y} we get the much simpler
(but still necessarily indeterminate) string
$$\s{x'} = a\{ab\}b\{ab\}b\{ab\},$$
whose associated graph $\mathcal{G}_{\s{x'}}$
has, in addition to the nine edges of $E_{\s{x}}$,
also the four (now positive) edges $(2,4),\,(2,6),\,(3,5),\,(4,6)$.
Thus in $\mathcal{G}_{\s{x'}}$ there are only two
maximal cliques, on the vertices 23456 and 1246, independent of each other,
and so by Lemma~\ref{lemm-minsig} \s{x'} can be constructed
using $\sigma = 2$ regular letters.

The fastest known algorithm to compute all maximal cliques
is described in \cite{BK73},
but of course it must be exponential in the worst case
($3^{n/3}$ maximal cliques);
it is not known how many independent maximal cliques can exist
in a graph constructed from a prefix array.
The graph $\mathcal{P}^+$ corresponding to $\s{y_2} = 80420311$
contains seven independent maximal cliques
$(138,146,17,24,25,27,35)$.
Thus, regarding this graph as an associated graph $\mathcal{G}_{\s{x}}$
of some string \s{x} tells us by Lemma~\ref{lemm-minsig}
that seven regular letters would be needed to represent it.

\section{Summary \& Future Work}
\label{sect-conc}
In this paper we have explored connections among
indeterminate strings, prefix arrays, and undirected graphs,
some of them quite unexpected (by us, at least).
We believe that many other connections exist
that may yield combinatorial insights and
thus more efficient algorithms.
For example:
\begin{enumerate}
\item
How many independent  maximal cliques can exist in the associated graph
$\mathcal{G}_{\s{x}}$ of a string \s{x} computed (on a minimum alphabet)
from a given prefix array \s{y}?
\item
Find an efficient algorithm to compute a string on a minimum alphabet
corresponding to a given nonregular prefix array.
\item
What classes of graphs $\mathcal{G}$ exist that,
as associated graphs $\mathcal{G} = \mathcal{G}_{\s{x}}$ of some string \s{x},
have fewer than exponential independent maximal cliques,
and so therefore may give rise to efficient algorithms for the
determination of \s{x} on a minimum alphabet?
Put another way: characterize graphs that have an exponential
number of independent maximal cliques.
\item
Can we recognize strings \s{x} with associated graphs $\mathcal{G}_{\s{x}}$
that have an exponential number of independent maximal cliques?
\item
Can known results from graph theory be used to design efficient
algorithms for computing patterns in indeterminate strings?
\end{enumerate}

\section*{Acknowledgements}
We are grateful to Jean-Pierre Duval and Arnaud Lefebvre
of the Universit\'{e} de Rouen for useful discussions.

\def\AJC{Australasian J.\ Combinatorics\ }
\def\AWOCA{Australasian Workshop on Combinatorial Algs.}
\def\CPM{Annual Symp.\ Combinatorial Pattern Matching}
\def\COCOON{Annual International Computing \& Combinatorics Conference}
\def\FOCS{IEEE Symp.\ Found.\ Computer Science}
\def\AESA{Annual European Symp.\ on Algs.}
\def\LATA{Internat.\ Conf.\ on Language \& Automata Theory \& Applications}
\def\IWOCA{Internat.\ Workshop on Combinatorial Algs.}
\def\AWOCA{Australasian Workshop on Combinatorial Algs.}
\def\STACS{Symp.\ Theoretical Aspects of Computer Science}
\def\ICALP{Internat.\ Colloq.\ Automata, Languages \& Programming}
\def\IJFCS{Internat.\ J.\ Foundations of Computer Science\ }
\def\ISAAC{Internat.\ Symp.\ Algs.\ \& Computation}
\def\SPIRE{String Processing \& Inform.\ Retrieval Symp.}
\def\SWAT{Scandinavian Workshop on Alg.\ Theory}
\def\PSC{Prague Stringology Conf.}
\def\ALG{Algorithmica\ }
\def\CSUR{ACM Computing Surveys\ }
\def\FI{Fundamenta Informaticae\ }
\def\IPL{Inform.\ Process.\ Lett.\ }
\def\IS{Inform.\ Sciences\ }
\def\JACM{J.\ Assoc.\ Comput.\ Mach.\ }
\def\CACM{Commun.\ Assoc.\ Comput.\ Mach.\ }
\def\MCS{Math.\ in Computer Science\ }
\def\NJC{Nordic J.\ Comput.\ }
\def\SICOMP{SIAM J.\ Computing\ }
\def\SIDMA{SIAM J.\ Discrete Math.\ }
\def\JCB{J.\ Computational Biology\ }
\def\JA{J.\ Algorithms\ }
\def\JCMCC{J.\ Combinatorial Maths.\ \& Combinatorial Comput.\ }
\def\JDA{J.\ Discrete Algorithms\ }
\def\JALC{J.\ Automata, Languages \& Combinatorics\ }
\def\SODA{ACM-SIAM Symp.\ Discrete Algs.\ }
\def\SPE{Software, Practice \& Experience\ }
\def\TCJ{The Computer Journal\ }
\def\TCS{Theoret.\ Comput.\ Sci.\ }


\begin{thebibliography}{WWW99}

\bibitem[A87]{A87} Karl Abrahamson,
{\bf Generalized string matching}, {\it \SICOMP 16--6} (1987)
1039--1051.
\bibitem[BIST03]{BIST03} H.\ Bannai, S.\ Inenaga, A.\ Shinohara \& M.\ Takeda,
{\bf Inferring strings from graphs and arrays},
{\em Mathematical Foundations of Computer Science},
Springer Lecture Notes in Computer Science LNCS 2747,
B.\ Rovan \& P.\ Vojt\'{a}s (eds.) (2003) 208--217.
\bibitem[B08]{B08} Francine Blanchet-Sadri,
{\em Algorithmic Combinatorics on Partial Words},
Chapman \& Hall/CRC (2008) 385 pp.
\bibitem[BSH02]{BSH02} Francine Blanchet-Sadri \& Robert A.\ Hegstrom,
{\bf Partial words and a theorem of Fine and Wilf revisited},
{\it \TCS 270--1/2} (2002) 401--409.
\bibitem[BM08]{BM08} J.\ A.\ Bondy \& U.\ S.\ R.\ Murty,
{\it Graph Theory}, Springer (2008) 651 pp.
\bibitem[CCR09]{CCR09} Julien Cl\'{e}ment, Maxime Crochemore \& Giuseppina Rindone,
{\bf Reverse engineering prefix tables},
{\it Proc.\ 26th \STACS},
Susanne Albers \& Jean-Yves Marion (eds.) (2009) 289--300.
\bibitem[BK73]{BK73} C.\ Bron \& J.\ Kerbosch,
{\bf Algorithm 457: finding all cliques of an undirected graph},
{\it Communications of the ACM 16--9} (1973) 575--577.
\bibitem[CHL01]{CHL01}
Maxime Crochemore, Christophe Hancart \& Thierry Lecroq,
{\it Algorithmique du Texte}, Vuibert (2001) 347 pp.
\bibitem[CHL07]{CHL07}
Maxime Crochemore, Christophe Hancart \& Thierry Lecroq,
{\it Algorithms on Strings}, Cambridge University Press (2007) 392 pp.
\bibitem[DLL05]{DLL05} Jean-Pierre Duval, Thierry Lecroq \& Arnaud Lefebvre,
{\bf Border array on a bounded alphabet},
{\em \JALC 10--1} (2005) 51--60.
\bibitem[FP74]{FP74} Michael J.\ Fischer \& Michael S.\ Paterson,
{\bf String-matching and other products},
{\it Complexity of Computation, Proc.\ SIAM-AMS 7} (1974) 113-125.
\bibitem[FLRS99]{FLRS99}
Frantisek Franek, Weilin Lu, P.\ J.\ Ryan, W.\ F.\ Smyth, Yu Sun \& Lu Yang,
{\bf Verifying a border array in linear time} (preliminary version),
{\em Proc.\ 10th \AWOCA}, School of Computing,
Curtin University of Technology (1999) 26--33.
\bibitem[FGLR02]{FGLR02}
Frantisek Franek, Shudi Gao, Weilin Lu, P.\ J.\ Ryan, W.\ F.\ Smyth, Yu Sun
\& Lu Yang,
{\bf Verifying a border array in linear time},
{\it \JCMCC 42} (2002) 223-236.
\bibitem[FS06]{FS06} Frantisek Franek \& W.\ F.\ Smyth,
{\bf Reconstructing a suffix array},
{\em \IJFCS 17--6} (2006) 1281--1295.
\bibitem[HS03]{HS03} Jan Holub \& W.\ F.\ Smyth,
{\bf Algorithms on indeterminate strings},
{\it Proc.\ 14th \AWOCA} (2003) 36--45.
\bibitem[HSW06]{HSW06} Jan Holub, W.\ F.\ Smyth \& Shu Wang,
{\bf Hybrid pattern-matching algorithms on indeterminate strings},
{\it London Algorithmics and Stringology 2006},
J. Daykin, M. Mohamed \& K. Steinhoefel (eds.),
King's College London Series {\it Texts in Algorithmics} (2006) 115--133.
\bibitem[HSW08]{HSW08} Jan Holub, W.\ F.\ Smyth \& Shu Wang,
{\bf Fast pattern-matching on indeterminate strings},
{\it \JDA 6--1} (2008) 37--50.
\bibitem[IMMP03]{IMMP03} Costas S.\ Iliopoulos, Manal Mohamed, Laurent Mouchard, Katerina G. Perdikuri, W.\ F.\ Smyth \& Athanasios K.\ Tsakalidis,
{\bf String regularities with don't cares},
{\it \NJC 10--1} (2003) 40--51.
\bibitem[K68]{K68} Joseph W.\ Kitchen Jr., {\it Calculus of One Variable}, Addison-Wesley (1968).
\bibitem[L05]{L05} M.\ Lothaire,
{\it Applied Combinatorics on Words},
Cambridge University Press (2005) 610 pp.
\bibitem[ML84]{ML84} Michael G.\ Main \& Richard J.\ Lorentz,
{\bf An $O(n\log n)$ algorithm for finding all repetitions in a string},
{\it \JA 5} (1984) 422--432.
\bibitem[MM65]{MM65} J.\ W.\ Moon \& L.\ Moser,
{\bf On cliques in graphs},
{\it Israel J.\ Math.\ 3} (1965) 23--28.
\bibitem[MSM99]{MSM99} Dennis Moore, W.\ F.\ Smyth \& Dianne Miller,
{\bf Counting distinct strings}, {\it Algorithmica 13--1} (1999) 1--13.
\bibitem[MP70]{MP70} James H.\ Morris \& Vaughan R.\ Pratt,
{\it A Linear Pattern-Matching Algorithm},
Tech.\ Rep.\ 40, University of California, Berkeley (1970).
\bibitem[M55]{M55} J.\ Mycielski,
{\bf Sur le colorage des graphes},
{\it Colloq.\ Math.\ 3} (1955) 161--162.
\bibitem[S03]{S03} Bill Smyth,
{\it Computing Patterns in Strings},
Pearson Addison-Wesley (2003) 423 pp.
\bibitem[SW08]{SW08} W.\ F.\ Smyth \& Shu Wang,
{\bf New perspectives on the prefix array},
{\em Proc. 15th \SPIRE},
Springer Lecture Notes in Computer Science LNCS 5280 (2008) 133--143.
\bibitem[SW09a]{SW09a} W.\ F.\ Smyth \& Shu Wang,
{\bf A new approach to the periodicity lemma on strings with holes},
{\it \TCS 410--43} (2009) 4295--4302.
\bibitem[SW09]{SW09} W.\ F.\ Smyth \& Shu Wang,
{\bf An adaptive hybrid pattern-matching algorithm on indeterminate strings},
{\it \IJFCS 20--6} (2009) 985--1004.
\end{thebibliography}
\end{document}